\documentclass{article}

\usepackage{orcidlink}
\usepackage{fancyhdr}
\usepackage{amsmath,amssymb}
\usepackage{amsfonts}
\usepackage{graphicx} 
\usepackage{verbatim}
\usepackage{pgfplots}
\pgfplotsset{compat=1.17}
\usetikzlibrary{positioning,arrows,calc,math,angles,quotes}
\usepackage{blochsphere}
\usepackage{etoolbox}

\usepackage[T1]{fontenc}
\usepackage[utf8]{inputenc}
\usepackage{lmodern}
\usepackage{microtype}
\usepackage{amsmath,amssymb,amsthm,mathtools,mathrsfs,bm}
\usepackage{booktabs,array,longtable,multirow}
\usepackage{enumitem}
\usepackage{hyperref}
\usepackage{float}
\usepackage{tikz}
\usetikzlibrary{arrows.meta,positioning,calc}
\usepackage{comment}
\usetikzlibrary{quantikz2}
\usepackage{braket}
\usepackage{comment}
\usepackage{mathtools}
\usetikzlibrary{matrix, arrows.meta}
\usepackage{amsthm}
\usepackage{fancyhdr}
\newtheorem{theorem}{Theorem}[section] 
\newtheorem{lemma}[theorem]{Lemma}     
\newtheorem{proposition}[theorem]{Proposition}

\theoremstyle{definition}

\theoremstyle{remark}
\newtheorem*{remark}{Remark}
\theoremstyle{conjecture}
\newtheorem*{conjecture}{Conjecture}

\newcommand{\Cl}{\mathcal{C}\ell}

\begin{document}

\title{Quantum Algorithm for Clifford Multiplication}

\author{%
Kagwe A. Muchane\orcidlink{0009-0008-6649-3975}%
\thanks{\textcopyright\ Kagwave.}
}

\date{}

\setcounter{footnote}{1}

\maketitle

\pagestyle{plain}

\begin{abstract}
Given two dense multivectors of the Clifford algebra $\mathcal{C}\ell(V, Q)$ with $N=2^{p+q}$ coefficients, the fastest known classical algorithms compute their geometric product in $O(N^{\omega/2})$ arithmetic operations, where $\omega$ denotes the matrix multiplication exponent. I show that, under amplitude encoding, a quantum computer executes the geometric product in $O(\operatorname{polylog} N)$ time, using logarithmic space with sublogarithmic circuit depth. This exponential speedup establishes Clifford multiplication as a quantum primitive, providing an efficient computational foundation for quantum geometric algorithms and
relativistic simulations.
\end{abstract}
\medskip


\newpage

\section{Introduction}

In his 1844 \textit{Ausdehnungslehre} \cite{grassmann1844}, Hermann Grassmann introduced the non-metric exterior product, later extending his theory to the central product, which incorporated metric notions of magnitude and projection \cite{grassmann1862}. William Kingdon Clifford then synthesized this dual product system with Hamilton's quaternions into a single associative engine by imposing the vector squaring relation \(v^2=Q(v)\), calling the resulting algebra \emph{geometric algebra} \cite{clifford1878applications}. By unifying the inner and exterior products into one associative multiplication, the geometric product has become a common language for Euclidean, projective, conformal, and spacetime geometry \cite{Dorst2007GACS,Lounesto2001}, as well as 21st-century applications such as computer vision, robotics, computer graphics, and geometric machine learning \cite{wareham2005applications, bayro2001geometric, bayro2001geometric_neural}. Recent Clifford-based neural architectures replace conventional tensor operations by geometric products, allowing network layers to respect the geometric symmetries of the data \cite{ruhe2023clifford, ruhe2023geometric, buchholz2008clifford, brehmer2023geometric}. Likewise, Hestenes' spacetime algebra reformulates relativistic mechanics and field theory, where Lorentz transformations, spinors, and Minkowski spacetime are all expressed intuitively through Clifford multiplication \cite{Hestenes1966,HestenesSobczyk1984}. While broadly applicable, dense geometric algebra computations have proven to be computationally demanding. The difficulty is fundamentally combinatorial. Because multiplying two dense multivectors can mean combining exponentially many blade pairs, Clifford multiplication has represented one of the principal computational bottlenecks for large-scale geometric algebra computations.

The most powerful quantum algorithms demonstrate that, on a quantum computer, harmonic structure becomes computational advantage. The quantum Fourier transform \cite{coppersmith1994approximate} changes states to a basis in which that harmony becomes visible, forming the foundation of algorithms such as Shor's factoring algorithm \cite{shor1997polynomial} and hidden subgroup algorithms \cite{kitaev1995quantum,lomont2004hidden}. In my most recent work, I discovered that Clifford's geometric product is not the sum of two independent products, but rather the harmonic decomposition of exchange under the action of a transposition $\tau\in S_2$. This replaces the traditional axiomatic view of the inner and exterior products as independent operations; they emerge as Fourier sectors of the exchange orbit of $\tau$. That finding naturally raised the following question: \textit{If the geometric product already possesses an intrinsic harmonic structure, can Clifford multiplication itself be realized as a quantum harmonic computation?}

The answer is \textit{yes}. Admittedly, this is possible because the geometric product can be seen as cocycle-twisted convolution over the group $(\mathbb Z_2)^n$. I formulate quantum Clifford multiplication as a harmonic computation whose underlying group operation is separated from the orientation information carried by the cocycle. Unlike operator-based quantum treatments of Clifford algebras, which execute Clifford multiplication through the action of quantum operators on states \cite{muchane2025stateoperator}, this algorithm computes the bilinear coefficient product without realizing the multivectors as operators. It computes dense Clifford multiplication with polylogarithmic gate complexity in the coefficient dimension, fitting neatly within the standard amplitude-encoding paradigm of quantum algorithms. The output is a quantum state representing the coefficient vector of the geometric product, from which observables of the product may be efficiently estimated.

The historical perception of Clifford algebra as a specialized mathematical tool has persisted despite its conceptual and interpretational advantages. One contributing factor is that, until recently, those advantages had not been accompanied by a computational advantage. This work permanently changes that perspective; the circuit developed here is universal, and independent of the input multivectors. Efficient quantum multiplication of multivectors opens the possibility of quantum-native algorithms for geometric machine learning, Clifford-valued signal processing, spacetime and relativistic simulation, geometric robotics, and other computational problems whose natural language is geometric algebra.
\section{Quantum Clifford Multiplication}
Dense Clifford multiplication is the function problem of computing the geometric product
\[
f:\mathcal C\ell_{p,q}\times\mathcal C\ell_{p,q}\rightarrow
\mathcal C\ell_{p,q},
\qquad
(A,B)\longmapsto AB.
\]
Every element in \( \mathcal{C}\ell_{p,q}\) is a linear combination \(A=\sum_{x\in I}a_xe_x\), where \(I\) is the set of all \(N=2^n\) basis blades and \(n=p+q\). Since \(|I|=N\), each dense multivector is specified by \(N\) coefficients. By bilinearity of the geometric product,
\begin{equation}
f(A,B)
=
\sum_{x\in I}\sum_{y\in I}
a_xb_y(e_xe_y).  
\end{equation}
There are two ways to measure the complexity of this problem. Treating the coefficient arrays themselves as the input, dense Clifford multiplication is polynomial in the input size \(N=2^n\) and therefore belongs to \(\mathsf{FP}\) under the standard bit-complexity model. Treating the geometric dimension \(n\) as the scaling parameter exposes the combinatorial growth of the algebra. Since \(N=2^n\), every known classical algorithm for dense multivectors requires exponential time in \(n\), placing the problem in \(\mathsf{FEXPTIME}\). Throughout this paper we adopt the latter viewpoint, as the uniform, representation-independent setting for comparing classical and quantum algorithms on families of Clifford algebras. For a classical computer, the product of two dense multivectors requires evaluating one interaction for every ordered pair of basis blades. Since an $n$-dimensional Clifford algebra contains $N=2^n$ basis blades, the straightforward coordinate expansion performs
\[
\Theta(N^2)=\Theta(4^n)
\]
blade interactions.

Modern classical improvements focus on reducing the cost of individual blade interactions or exploiting alternative representations. Basis blades are encoded as bitstrings, output blades are computed by bitwise XOR, sign is evaluated using parity operations, and SIMD or GPU implementations evaluate many blade products simultaneously
\cite{Ablamowicz2009Complexity,Dorst2007GACS}. These optimizations significantly improve constant factors and practical performance, but they do not alter the asymptotic complexity of dense multiplication, as every nonzero coefficient of the first multivector may still interact with every nonzero coefficient of the second. Consequently, dense Clifford multiplication retains quadratic complexity in the coefficient dimension, or equivalently, exponential complexity in the geometric dimension \(n\). The fastest known asymptotic algorithms instead exploit faithful matrix representations of finite-dimensional Clifford algebras, thereby reducing dense Clifford multiplication to matrix multiplication. The best current
classical matrix multiplication algorithms run in time
$
O\!\left(N^{\omega/2}\right),
$ with the current optimal upper bound being \(\omega < 2.371339\) \cite{AlmanWilliams2021} for the matrix multiplication exponent. This improves the dependence on the coefficient dimension from quadratic to
$
O(N^{1.186}),
$
but the computation remains exponential in the geometric dimension,
\[
O\!\left((2^n)^{\omega/2}\right)
=
O\!\left(2^{(\omega/2)n}\right).
\]

The quantum algorithm presented in this paper follows a different approach, exploiting the algebraic structure of the geometric product itself. I present an efficient quantum circuit whose complexity scales linearly with the geometric dimension, with the precise gate and depth bounds established in Theorem~\ref{qcm}. To expose the algebra responsible for this speedup, we first examine the multiplication law of the basis blades. Their binary encoding separates Clifford multiplication into a group operation on blade indices together with a cocycle encoding orientation.

\paragraph{Basis encoding.}
Every basis blade is uniquely determined by a subset of the basis vectors and is therefore identified with its characteristic bit vector. Writing
\[
x=\sum_{i=1}^n x_i2^{\,i-1},\qquad x_i\in\{0,1\},
\]
we identify the integer bitmask \(x\) with the blade
\[
e_x=e_1^{x_1}e_2^{x_2}\cdots e_n^{x_n}.
\]
Throughout, bitstrings are displayed in the conventional most-significant-bit-first
order \((x_n,\ldots,x_1)\). The displayed bitstring records the presence of generators \((e_n,\ldots,e_1)\), while the integer encoding is little-endian.

\begin{table}[H]
\centering
\caption{Binary encoding of basis blades together with representative products in
$\Cl_{3,0}(\mathbb R)$. The output blade is determined by bitwise XOR, while
the Clifford cocycle contributes the orientation sign.}
\label{tab:cl3_encoding_products}
\begin{tabular}{cc|cc|ccc}
\toprule
$e_x$ & $x$ & $e_y$ & $y$ & $x\oplus y$ & $\chi_{3,0}(x,y)$ & $e_xe_y$ \\
\midrule
$1$        & 000 & $1$        & 000 & 000 & $+1$ & $1$ \\
$e_1$      & 001 & $e_2$      & 010 & 011 & $+1$ & $e_{12}$ \\
$e_2$      & 010 & $e_{13}$   & 101 & 111 & $-1$ & $-e_{123}$ \\
$e_{12}$   & 011 & $e_{123}$  & 111 & 100 & $-1$ & $-e_3$ \\
$e_3$      & 100 & $e_1$      & 001 & 101 & $-1$ & $-e_{13}$ \\
$e_{13}$   & 101 & $e_{12}$   & 011 & 110 & $+1$ & $e_{23}$ \\
$e_{23}$   & 110 & $e_3$      & 100 & 010 & $+1$ & $e_2$ \\
$e_{123}$  & 111 & $e_{23}$   & 110 & 001 & $-1$ & $-e_1$ \\
\bottomrule
\end{tabular}
\end{table}

The geometric product of basis blades separates into an output index and a
phase:
\[
e_xe_y=\chi_{p,q}(x,y)e_{x\oplus y},
\]
where $\oplus$ denotes bitwise XOR and
$\chi_{p,q}(x,y)\in\{\pm1\}$ is the Clifford cocycle. The XOR operation determines the output blade, while the Clifford cocycle encodes the orientation information from basis exchanges together with the metric signs determined by the signature $(p,q)$.

\subsubsection{Twisted Convolution}

In recent work on generalized Clifford geometry \cite{Muchane2026Rotor}, I showed that the geometric product can be understood through a harmonic decomposition under the exchange action of transposition $\tau$. Though the full theory is developed elsewhere, it suffices to note that the $\mathbb{Z}_2$ Fourier transform diagonalizes the exchange operator $\tau$, whose spectral projectors $\frac{1}{2}(1 \pm \tau)$ separate an ordered product $ab$ into its symmetric and antisymmetric exchange sectors. For two-vectors, the exchange orbit is $(ab,\tau(ab))=(ab,ba)$, and the normalized character table of $\mathbb{C}[\mathbb{Z}_2]$ acts by
\[
F_2
\begin{pmatrix}
ab\\
\tau(ab)
\end{pmatrix}
=
\frac{1}{\sqrt{2}}
\begin{pmatrix}
1&1\\
1&-1
\end{pmatrix}
\begin{pmatrix}
ab\\
ba
\end{pmatrix}
=
\frac{1}{\sqrt{2}}
\begin{pmatrix}
ab+ba\\
ab-ba
\end{pmatrix}
=
\sqrt2
\begin{pmatrix}
a\cdot b\\
a\wedge b
\end{pmatrix}.
\]
Thus, up to normalization, the orthogonal Clifford decomposition is the $\mathbb{Z}_2$ Fourier decomposition of the exchange orbit. The same exchange decomposition naturally extends from vectors to arbitrary basis blades and, by linearity, to dense multivectors.

For basis blades, every product \(e_xe_y\) determines an exchange orbit
\((e_xe_y,e_ye_x)\), whose relative orientation is encoded by the Clifford
cocycle. The exchange rotor introduced in~\cite{Muchane2026Rotor} is defined
by
$
R=(ba)^{-1}ab
$.
For basis blades \(a=e_x\) and \(b=e_y\), the geometric product is
\[
ab=\chi_{p,q}(x,y)e_{x\oplus y},
\qquad
ba=\chi_{p,q}(y,x)e_{x\oplus y}.
\]
Both orderings produce the same output blade \(e_{x\oplus y}\), differing only in their cocycle values. Since every nonzero basis blade is invertible,
\[
\begin{aligned}
R
&=(ba)^{-1}ab\\
&=\left(\chi_{p,q}(y,x)e_{x\oplus y}\right)^{-1}
   \left(\chi_{p,q}(x,y)e_{x\oplus y}\right)\\
&=\frac{\chi_{p,q}(x,y)}{\chi_{p,q}(y,x)}
   e_{x\oplus y}^{-1}e_{x\oplus y}\\
&=\frac{\chi_{p,q}(x,y)}{\chi_{p,q}(y,x)}.
\end{aligned}
\]
The Clifford cocycle therefore completely determines the multiplication law of the basis blades. By contrast, the exchange rotor measures only the relative effect of reversing the order of multiplication, which is the ratio between the two cocycle values. This distinction becomes clearer by decomposing the cocycle into exchange and metric contributions. The sign in Clifford multiplication has two independent contributions: the parity of transpositions required to reorder basis vectors, and the signature-dependent signs from the contractions $e_i^2=\pm1$. Since these contributions occur independently during multiplication, the cocycle factors as
\[
\chi_{p,q}(x,y)
=
\chi_{\mathrm{ex}}(x,y)
\chi_{\mathrm{met}}(x,y).
\]
\(\chi_{\mathrm{ex}}\) records the signs introduced when restoring canonical blade order, while \(\chi_{\mathrm{met}}\) records the metric signs from repeated generators. Since the metric contribution is symmetric and
$
\chi_{\mathrm{met}}(x,y)
=
\chi_{\mathrm{met}}(y,x),
$
it cancels in the ratio, giving
\[
R
=
\frac{\chi_{\mathrm{ex}}(x,y)}
     {\chi_{\mathrm{ex}}(y,x)}.
\]
The exchange rotor depends only on the exchange (orientation) component of the cocycle. 

By linearity, dense Clifford multiplication is a
superposition of such decorated exchange orbits. For dense multivectors $A=\sum_xa_xe_x$ and $B=\sum_yb_ye_y$,
\[
AB
=
\sum_{x,y}
a_xb_y
\chi_{p,q}(x,y)
e_{x\oplus y}.
\]
Grouping together those having the same output blade \(z=x\oplus y\) gives
\[
AB
=
\sum_zc_ze_z,
\qquad
c_z
=
\sum_{x\oplus y=z}
a_xb_y\chi_{p,q}(x,y)
=
\sum_x
a_xb_{x\oplus z}
\chi_{p,q}(x,x\oplus z).
\]
The coefficient map
$
c_z
=
\sum_{x\oplus y=z}
a_xb_y\chi_{p,q}(x,y)
$
is a cocycle-twisted convolution over
$(\mathbb Z_2)^n$
\cite{AlbuquerqueMajid2002,Lounesto2001}.
For a Clifford algebra \(\mathcal{C}\ell_{p,q}(\mathbb K)\), with
\(\mathbb K\in\{\mathbb R,\mathbb C\}\), one has
\[
\mathcal{C}\ell_{p,q}(\mathbb K)
\cong
\mathbb K_{\chi}\!\left[(\mathbb Z_2)^n\right],
\qquad n=p+q,
\]
where the cocycle \(\chi\) is determined by the quadratic form \(Q\) after choosing an orthogonal basis. Ordinary convolution on the Abelian group \((\mathbb Z_2)^n\) is diagonalized by the Walsh--Hadamard transform, the character table of \((\mathbb Z_2)^n\). Clifford multiplication differs only
by the additional cocycle factor
\[
(f *_\chi g)(z)
=
\sum_x
f(x)g(x\oplus z)\chi(x,x\oplus z).
\]

Passing to complex scalars does not alter the underlying twisted group algebra. 
\[
\mathbb C\otimes_{\mathbb R}\mathcal{C}\ell(V,Q)
\cong
\mathbb C_{\chi}\!\left[(\mathbb Z_2)^n\right],
\]
so the blade-index group, cocycle, and twisted convolution are unchanged;
only the coefficient field is extended. Since quantum amplitudes are already
complex, the quantum circuit, including the cocycle phase oracle, covers complexified Clifford algebra without changing the asymptotic resource bounds. This reformulation is the key to the quantum algorithm. Rather than evaluating all \(4^n\) blade interactions individually, the circuit computes the cocycle-twisted convolution coherently in amplitude space by combining reversible XOR, a diagonal cocycle phase oracle, and a Walsh--Hadamard transform on the summation index.

\subsection{Algorithm}
The input to the algorithm is two amplitude-encoded multivectors, where each computational basis state $\ket{x}$ corresponds to the Clifford basis blade $e_x$. The algorithm is executed by a universal quantum circuit $\mathcal U$ that prepares an amplitude encoding of the geometric product
$AB=\sum_z c_ze_z,$
that is,
\begin{equation}
\mathcal U:
\left(
|A\rangle=\sum_xa_x|x\rangle,\;
|B\rangle=\sum_yb_y|y\rangle
\right)
\longmapsto
|AB\rangle=\sum_zc_z|z\rangle,
\end{equation}
where the coefficients $\{c_z\}$ are those of the geometric product $AB$. Because the output remains amplitude encoded, the multiplication primitive can be applied repeatedly, and chains of geometric products may be evaluated coherently within a larger quantum computation before a final measurement is performed.

\subsubsection{The cocycle oracle}
Having reduced Clifford multiplication to cocycle-twisted convolution, whose group operation is bitwise XOR, the remaining challenge is the implementation of the cocycle. To this end, I design a quantum oracle $U_\chi$ for its phase. For $\Cl_{p,q}$, the Clifford cocycle is encoded by the Boolean phase polynomial
\[
\Phi_{p,q}:(\mathbb Z_2)^n\times(\mathbb Z_2)^n\longrightarrow \mathbb Z_2,
\]
which records the Boolean exponent of the cocycle. With the little-endian blade encoding fixed above,
$
e_xe_y=\chi_{p,q}(x,y)e_{x\oplus y},
$
where
\[
\chi_{p,q}(x,y)=(-1)^{\Phi_{p,q}(x,y)}
\]
and
\[
\Phi_{p,q}(x,y)
=
\sum_{1\le i<j\le n}x_jy_i
+
\sum_{i=p+1}^{n}x_iy_i
\pmod 2.
\]

The phase polynomial separates into two independent contributions. The first term records the parity of the exchanges needed to restore
canonical blade order: a generator $e_j$ from the left factor crosses a generator $e_i$ from the right factor whenever $i<j$. The second term
records the metric contribution from repeated generators in the negative directions, where $e_i^2=-1$. Thus the exchange term is purely combinatorial, depending only on the ordering of basis generators, while the metric term depends only on the signature $(p,q)$ of the defining
quadratic form. 
Since
\[
\Phi_{p,q}
=
\Phi_{\mathrm{ex}}
+
\Phi_{\mathrm{met}}
\pmod2,
\]
the Clifford cocycle factorizes as
$
\chi_{p,q}(x,y)
=
\chi_{\mathrm{ex}}(x,y)
\chi_{\mathrm{met}}(x,y),
$
where
\[
\chi_{\mathrm{ex}}(x,y)
=
(-1)^{\Phi_{\mathrm{ex}}(x,y)},
\qquad
\chi_{\mathrm{met}}(x,y)
=
(-1)^{\Phi_{\mathrm{met}}(x,y)}.
\]
The oracle mirrors this decomposition by first computing the Boolean exponent of the exchange cocycle through a reversible linear transformation, then using phase kickback to implement the map
$
\Phi_{\mathrm{ex}}(x,y)
\longrightarrow
(-1)^{\Phi_{\mathrm{ex}}(x,y)}.
$
An additional diagonal phase layer contributes the metric phase, so that
\[
\chi_{p,q}(x,y)
=
\chi_{\mathrm{ex}}(x,y)\chi_{\mathrm{met}}(x,y)
=
(-1)^{\Phi_{\mathrm{ex}}(x,y)+\Phi_{\mathrm{met}}(x,y)}
=
(-1)^{\Phi_{p,q}(x,y)}.
\]

A direct implementation of the exchange phase would apply a controlled-$Z$ gate for every pair \((i,j)\) with \(i<j\). Since there are \(\binom{n}{2}\) such pairs, this naive oracle uses \(O(n^2)\) two-qubit gates to prepare
$
(-1)^{\sum_{i<j}x_jy_i}.
$
The metric contribution is diagonal and requires only \(q\) controlled-$Z$ gates. Although this already reduces the dependence from the exponential coefficient dimension to the geometric dimension, the exchange phase order can be synthesized substantially more efficiently.

\begin{lemma}[Nilpotent shift representation of the exchange polynomial]
Let $S\in M_n(\mathbb F_2)$ be the lower nilpotent shift matrix,
\[
S=
\begin{pmatrix}
0&0&0&\cdots&0\\
1&0&0&\cdots&0\\
0&1&0&\cdots&0\\
\vdots&&\ddots&\ddots&\vdots\\
0&0&\cdots&1&0
\end{pmatrix}
\qquad
S^n=0.
\]
 Then the exchange contribution to the Boolean phase polynomial satisfies
\[
\sum_{1\le i<j\le n}x_jy_i
=
x^TLy,
\]
where
\[
L=\sum_{k=1}^{n-1}S^k.
\]
\end{lemma}

\begin{proof}
For each $k\ge1$, the matrix $S^k$ has ones on the $k$-th subdiagonal.
\[
L=\sum_{k=1}^{n-1}S^k
=
\begin{pmatrix}
0&0&0&\cdots&0\\
1&0&0&\cdots&0\\
1&1&0&\cdots&0\\
\vdots&&\ddots&\ddots&\vdots\\
1&1&\cdots&1&0
\end{pmatrix},
\]
so
\[
x^TLy
=
\sum_{1\le i<j\le n}x_jy_i.
\]
The bilinear form \(x^TLy\) therefore computes the exchange contribution \(\Phi_{\mathrm{ex}}(x,y)\) to the Boolean phase polynomial. Now define
$
\mathbf T:=(I+S)^{-1}.
$
Since
$
(I+S)(I+S+\cdots+S^{n-1})=I,
$
we have
\[
\mathbf T
=
I+S+\cdots+S^{n-1}
=
I+L.
\]
Unlike an infinite Neumann series, this resolvent expansion terminates exactly because $S$ is nilpotent. Moreover, since subtraction equals addition in $\mathbb F_2$,
\[
L=\mathbf T-I=\mathbf T+I.
\]
\end{proof}

\begin{remark}
$\mathbf T$ has three equivalent interpretations:

(i) \textit{the finite resolvent of the nilpotent shift operator},

(ii) \textit{a lower-triangular Toeplitz matrix over $\mathbb F_2$}, and

(iii) \textit{the binary prefix-sum transform}.
\end{remark}

Viewed as the binary prefix-sum transform, $\mathbf T$ is an accumulation operator. The shift $S$ advances one position through the ordered blade indices, while the finite series $I+S+S^2+\cdots+S^{n-1}$ accumulates the exchange contributions required to
restore canonical blade order. The resolvent form of $\mathbf T$ admits two complementary factorizations. The first is an elementary factorization into nearest-neighbor transvections. The second is a dyadic factorization into powers of the nilpotent shift.

\begin{lemma}[Elementary sweep factorization]
Let $e_{i+1,i}\in M_n(\mathbb F_2)$ denote the matrix with a single nonzero
entry in position $(i+1,i)$, and define
\[
E_{i+1,i}:=I+e_{i+1,i}.
\]
Then
\[
\mathbf T
=
E_{n,n-1}E_{n-1,n-2}\cdots E_{2,1}.
\]
\end{lemma}
\begin{proof}
The elementary matrix $E_{i+1,i}$ acts on column vectors by adding coordinate
$i$ into coordinate $i+1$,
\[
y_{i+1}\longmapsto y_{i+1}+y_i,
\]
leaving all other coordinates fixed. Since the product acts on column
vectors from right to left, the factorization applies the sweep
\[
1\to2,\quad 2\to3,\quad \ldots,\quad n-1\to n.
\]
After this sweep, the $j$-th output coordinate is
$
y_1+y_2+\cdots+y_j.
$
Therefore the product has entries
\[
\left(E_{n,n-1}\cdots E_{2,1}\right)_{ji}
=
\begin{cases}
1,& i\le j,\\
0,& i>j,
\end{cases}
\]
which is exactly the unit lower-triangular matrix $\mathbf T$.
\end{proof}

\begin{lemma}[Dyadic factorization]
Let
$
h=\lceil \log_2 n\rceil.
$
Then
\[
\mathbf T
=
\prod_{m=0}^{h-1}
\left(I+S^{2^m}\right).
\]
\end{lemma}
\begin{proof}
First suppose \(n=2^r\). Define
$
P_r:=\prod_{m=0}^{r-1}\left(I+S^{2^m}\right).
$
We prove by induction that
$
P_r=I+S+S^2+\cdots+S^{2^r-1}.
$
For \(r=1\), this is immediate. If the identity holds for \(r\), then
\[
\begin{aligned}
P_{r+1}
&=P_r\left(I+S^{2^r}\right)\\
&=\left(\sum_{k=0}^{2^r-1}S^k\right)
+
\left(\sum_{k=0}^{2^r-1}S^{k+2^r}\right)\\
&=\sum_{k=0}^{2^{r+1}-1}S^k.
\end{aligned}
\]
For general \(n\), take \(h=\lceil\log_2 n\rceil\). Since \(S^n=0\), all
terms \(S^k\) with \(k\ge n\) vanish, so the identity truncates to
\[
\prod_{m=0}^{h-1}
\left(I+S^{2^m}\right)
=
\sum_{k=0}^{n-1}S^k
=
\mathbf T.
\]
\end{proof}

\begin{proposition}[Optimized Clifford cocycle oracle]
The Clifford cocycle phase oracle
\[
U_\chi : \ket{x}\ket{y}
\longmapsto
(-1)^{\Phi_{p,q}(x,y)}\ket{x}\ket{y}
\]
can be synthesized in two ways. The elementary sweep uses $O(n)$ two-qubit gates and depth $O(n)$. The dyadic method uses $O(n\log n)$ two-qubit gates and depth $O(\log n)$.
\end{proposition}

\begin{proof}
Both constructions follow from factorizations of the nilpotent resolvent
\(
\mathbf T=(I+S)^{-1}.
\)
Since $L=\mathbf T+I$ over $\mathbb F_2$, we have
$
x^TLy=x^T\mathbf Ty+x^Ty.
$
Therefore
\[
\Phi_{p,q}(x,y)
=
x^T\mathbf Ty
+
x^Ty
+
\sum_{i=p+1}^{n}x_iy_i.
\]
The last two terms combine over $\mathbb F_2$,
\[
x^Ty+\sum_{i=p+1}^{n}x_iy_i
=
\sum_{i=1}^{p}x_iy_i.
\]
Thus the nontrivial exchange contribution is implemented through
\[
U_{\mathbf T}:\ket{y}\longmapsto\ket{\mathbf Ty}.
\]
Since \(\mathbf T\in GL_n(\mathbb F_2)\), the map \(y\mapsto \mathbf T y\) is a bijection of \((\mathbb Z_2)^n\), and \(U_{\mathbf T}\ket{y}=\ket{\mathbf T y}\) is a unitary permutation of the computational basis.  A parallel controlled-$Z$ layer between corresponding qubits of $\ket{x}$ and $\ket{\mathbf Ty}$ kicks back the phase
$
(-1)^{x^T\mathbf Ty},
$
since the $i$-th controlled-$Z$ gate contributes
$(-1)^{x_i(\mathbf Ty)_i}$, and the phases multiply to
\[
(-1)^{\sum_i x_i(\mathbf Ty)_i}
=
(-1)^{x^T\mathbf Ty}.
\]
Applying $U_{\mathbf T}^{\dagger}$ restores the second register to
$\ket{y}$ while preserving this phase. A final controlled-$Z$ layer over
the positive directions contributes
$(-1)^{\sum_{i=1}^{p}x_iy_i}$. Using \(L=\mathbf T+I\) and the fact that addition and subtraction coincide
over \(\mathbb F_2\), we may rewrite
\[
x^T\mathbf Ty
=
x^T(\mathbf T+I)y+x^Ty.
\]
Hence the total phase is
\[
\begin{aligned}
(-1)^{x^T\mathbf Ty+\sum_{i=1}^{p}x_iy_i}
&=
(-1)^{
x^T(\mathbf T+I)y
+
x^Ty
+
\sum_{i=1}^{p}x_iy_i
}\\
&=
(-1)^{
x^T(\mathbf T+I)y
+
\sum_{i=p+1}^{n}x_iy_i
}\\
&=
(-1)^{\Phi_{p,q}(x,y)}.
\end{aligned}
\]
\end{proof}

\begin{table}[t]
\centering
\begin{tabular}{lccc}
\hline
Implementation & Gates & Qubits & Depth \\
\hline
Naive & $O(\log^2 N)$ & $O(\log N)$ & $O(\log^2 N)$ \\
Sweep & $O(\log N)$ & $O(\log N)$ & $O(\log N)$ \\
Dyadic prefix & $O(\log N\log\log N)$ & $O(\log N)$ & $O(\log\log N)$ \\
\hline
\end{tabular}
\caption{Circuit resources for three implementations of the Clifford cocycle phase oracle \(U_\chi\), where \(N=2^n\) is the number of Clifford basis blades. Gate count measures elementary Clifford gates, qubit count measures the two input registers and any working space, and depth assumes parallel execution of disjoint gates.}
\label{tab:circuit_resources}
\end{table}

\begin{figure}[H]
\centering
\begin{quantikz}[row sep=0.35cm,column sep=0.55cm]
\lstick{$\ket{x}$}
    & \qw
    & \gate[wires=2]{\mathrm{CZ}_{\mathbf T}}
    & \qw
    & \gate[wires=2]{\mathrm{CZ}_{+}}
    & \qw
\\
\lstick{$\ket{y}$}
    & \gate{U_{\mathbf T}}
    & \qw
    & \gate{U_{\mathbf T}^{\dagger}}
    & \qw
    & \qw
\end{quantikz}
\caption{Optimized Clifford cocycle oracle $U_\chi$. The transform
$U_{\mathbf T}$ maps $\ket{y}$ to $\ket{\mathbf T y}$. The layer
$\mathrm{CZ}_{\mathbf T}$ produces the phase $(-1)^{x^T\mathbf T y}$. The inverse transform uncomputes the second register, and $\mathrm{CZ}_{+}$ contributes the remaining diagonal phase over the positive directions.}
\label{fig:cocycle_oracle}
\end{figure}
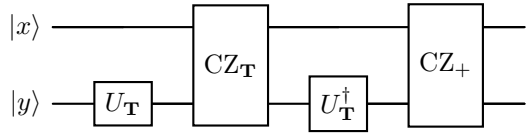

\begin{figure}[H]
\centering
\begin{quantikz}[row sep=0.22cm,column sep=0.34cm]
\lstick{$|y_1\rangle$} & \ctrl{1} & \qw      & \qw      & \qw      & \qw \\
\lstick{$|y_2\rangle$} & \targ{}  & \ctrl{1} & \qw      & \qw      & \qw \\
\lstick{$|y_3\rangle$} & \qw      & \targ{}  & \ctrl{1} & \qw      & \qw \\
\lstick{$|y_4\rangle$} & \qw      & \qw      & \targ{}  & \ctrl{1} & \qw \\
\lstick{$|y_5\rangle$} & \qw      & \qw      & \qw      & \targ{}  & \qw
\end{quantikz}
\caption{Elementary sweep variant of \(U_{\mathbf T}\) for \(n=5\).
Each CNOT implements the transvection \(E_{i+1,i}\), propagating the binary
prefix sum through the register. The circuit maps
\(\ket{y}\mapsto\ket{\mathbf T y}\) with \(n-1\) CNOT gates and depth
\(n-1\).}
\label{fig:sweep_prefix}
\end{figure}
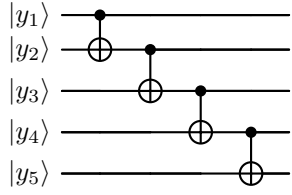

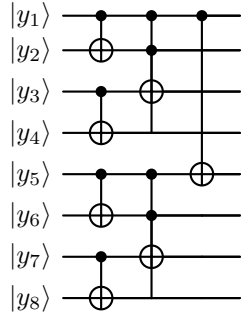
\begin{figure}[H]
\centering
\begin{quantikz}[row sep=0.22cm,column sep=0.34cm]
\lstick{$|y_1\rangle$} & \ctrl{1} & \ctrl{2} & \ctrl{4} & \qw \\
\lstick{$|y_2\rangle$} & \targ{}  & \ctrl{2} & \qw      & \qw \\
\lstick{$|y_3\rangle$} & \ctrl{1} & \targ{}  & \qw      & \qw \\
\lstick{$|y_4\rangle$} & \targ{}  & \qw      & \qw      & \qw \\
\lstick{$|y_5\rangle$} & \ctrl{1} & \ctrl{2} & \targ{}  & \qw \\
\lstick{$|y_6\rangle$} & \targ{}  & \ctrl{2} & \qw      & \qw \\
\lstick{$|y_7\rangle$} & \ctrl{1} & \targ{}  & \qw      & \qw \\
\lstick{$|y_8\rangle$} & \targ{}  & \qw      & \qw      & \qw
\end{quantikz}
\caption{Schematic dyadic prefix construction for \(U_{\mathbf T}\). The \(m\)-th stage propagates partial prefix sums across distance \(2^m\).
With a parallel-prefix implementation, the propagation distance doubles at
each stage, giving depth \(O(\log n)\) and \(O(n\log n)\) CNOT gates.}
\label{fig:dyadic_prefix}
\end{figure}

For the sweep, \(U_{\mathbf T}\) uses \(n-1\) CNOT gates, and
\(U_{\mathbf T}^{\dagger}\) uses another \(n-1\). The central
\(\mathrm{CZ}_{\mathbf T}\) layer consists of \(n\) controlled-\(Z\) gates,
and the final positive-direction diagonal layer consists of \(p\)
controlled-\(Z\) gates. Hence the sweep oracle uses
\[
(n-1)+(n-1)+n+p=3n-2+p
\]
two-qubit gates. Since \(0\le p\le n\), this is between \(3n-2\) and \(4n-2\), so both its gate count and depth are \(O(n)\). The dyadic variant requires
$
O(n\log n)
$
two-qubit gates and depth \(O(\log n)\). The sweep minimizes gate count, while the dyadic method minimizes circuit depth. Since the geometric dimension satisfies \(n=\log_2 N\), where \(N\) is the number of Clifford basis blades, the dyadic oracle has sublogarithmic depth
$
O(\log n)=O(\log\log N),
$
and gate complexity
$
O(n\log n)=O(\log N\,\log\log N).
$

\subsubsection{Procedure.}
The following section combines the oracle with
reversible XOR and the Walsh--Hadamard transform to complete the algorithm for Clifford multiplication.

\begin{theorem}

Let $A,B\in\mathcal C\ell_{p,q}(\mathbb K)$, with coefficient states $|A\rangle=\sum_xa_x|x\rangle$ and $|B\rangle=\sum_yb_y|y\rangle$. There exists an $O(n)$ circuit whose trivial-character branch prepares the normalized state proportional to $\sum_zc_z|z\rangle$. The success probability of this branch is
$
p_0=2^{-n}\sum_z|c_z|^2.
$
\label{qcm}
\end{theorem}

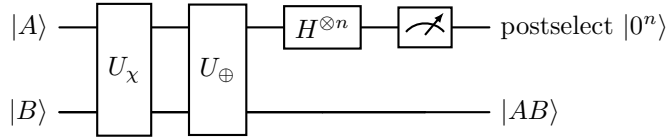
\begin{figure}[H]
\centering
\begin{quantikz}
\lstick{$|A\rangle$}
& \gate[2]{U_\chi}
& \gate[2]{U_{\oplus}}
& \gate{H^{\otimes n}}
& \meter{}
& \rstick{postselect $|0^n\rangle$}
\\
\lstick{$|B\rangle$}
& \qw
& \qw
& \qw
& \qw
& \rstick{$|AB\rangle$}
\end{quantikz}
\caption{Quantum Clifford multiplication circuit.}
\end{figure}

\begin{proof}
Starting from the tensor-product state
$
|A\rangle|B\rangle
=
\sum_{x,y}
a_xb_y
|x\rangle\otimes|y\rangle,
$
apply the cocycle phase oracle to the two registers
\[
U_{\chi}:
|x\rangle|y\rangle
\longmapsto
(-1)^{\Phi_{p,q}(x,y)}
|x\rangle|y\rangle.
\]
The oracle depends only on the Clifford multiplication law and is independent of the coefficients of \(A\) and \(B\). Next apply the reversible XOR map
\[
U_\oplus:
|x\rangle|y\rangle
\longmapsto
|x\rangle|x\oplus y\rangle,
\]
which computes the underlying group operation on blade indices. Writing
\(z=x\oplus y\), equivalently \(y=x\oplus z\), the state becomes,
\[
|\psi\rangle
=
\sum_{x,z}
a_x
b_{x\oplus z}
\chi_{p,q}(x,x\oplus z)
|x\rangle|z\rangle.
\]
Now apply the Walsh-Hadamard transform \(H^{\otimes n}\) to the first register (the \(A\)-register).
\[
H^{\otimes n}|x\rangle
=
2^{-n/2}
\sum_{w\in(\mathbb Z_2)^n}
(-1)^{w\cdot x}|w\rangle,
\]
we obtain
\[
|\psi\rangle
=
2^{-n/2}
\sum_{w,z}
\left(
\sum_x
(-1)^{w\cdot x}
a_x
b_{x\oplus z}
\chi_{p,q}(x,x\oplus z)
\right)
|w\rangle|z\rangle.
\]
The trivial-character branch \(w=0^n\) is therefore
\[
2^{-n/2}
|0^n\rangle
\sum_z
\left(
\sum_x
a_x
b_{x\oplus z}
\chi_{p,q}(x,x\oplus z)
\right)
|z\rangle
=
2^{-n/2}
|0^n\rangle
\sum_z
c_z
|z\rangle,
\]
where the inner sum is exactly the \(z\)-th coefficient of the Clifford product \(AB\).
Conditioned on this outcome, the second register is the normalized amplitude encoding of the geometric product \(AB\). The success
probability is the squared norm of the trivial-character branch,
\[
p_0
=
\left\|
2^{-n/2}
\sum_z
c_z
|z\rangle
\right\|^2
=
2^{-n}
\sum_z
|c_z|^2.
\]

The reversible XOR uses \(O(n)\) CNOT gates, and the Walsh--Hadamard transform uses \(n\) Hadamard gates. Thus the circuit has size \[ O(n)+\operatorname{cost}(U_\chi). \] Using the linear-size sweep cocycle oracle gives an \(O(n)\)-gate circuit; using the dyadic oracle gives \(O(n\log n)\). 
\end{proof}

\begin{remark}[Twisted group algebra form]
The proof of Theorem~\ref{qcm} depends only on four primitives:
\begin{enumerate}
    \item [(i)] a basis $\{e_g:g\in \mathcal{G} \}$ indexed by a finite group,
    \item [(ii)] an efficient reversible implementation of the group operation,
    \item [(iii)] a twisted multiplication $
e_ge_h=\chi_\mathcal{G}(g,h)e_{gh},
$
    \item [(iv)] a quantum Fourier transform for $\mathcal{G}$, enabling postselection onto the trivial representation.
\end{enumerate}
The proof is therefore independent of the specific Clifford algebra and should apply verbatim to any efficiently presented twisted group algebra
$
\mathbb C_\chi[\mathcal{G}].
$
Let $\chi=\chi_\mathcal{G}$ be a normalized $U(1)$-valued 2-cocycle on $\mathcal{G}$. For
$
A=\sum_ga_ge_g, B=\sum_hb_he_h,
$
the product coefficients are
$
c_k=\sum_{g\in \mathcal{G}}a_gb_{g^{-1}k}\chi_\mathcal{G}(g,g^{-1}k).
$
After applying the cocycle phase, the reversible multiplication map
$
|g\rangle|h\rangle\mapsto |g\rangle|gh\rangle,
$
and the quantum Fourier transform on the first register, postselection on the trivial representation prepares
$
|\mathcal{G}|^{-1/2}\sum_kc_k|k\rangle.
$
The Clifford case is $\mathcal{G}=(\mathbb Z_2)^n$, where $gh=g\oplus h$ and the Fourier transform is $H^{\otimes n}$. In general, the cost per attempt is driven by the reversible group multiplication, the cocycle phase oracle, and the quantum Fourier transform over $\mathcal{G}$, with any additional overhead determined by the success probability of the trivial-representation branch.
\end{remark}

This method differs from recent quantum algorithms for group convolution \cite{Castelazo2021QuantumConvolution}, which usually treat convolution as a linear operator determined by a fixed filter. Here the circuit carries out the bilinear multiplication of a twisted group algebra. The cocycle is incorporated explicitly through the diagonal phase oracle \(U_{\chi}\), while the Walsh--Hadamard transform performs the harmonic summation over the blade-index group.

\begin{conjecture}[Cocycle degree and Clifford hierarchy] Let \(\mathcal G\) be a finite abelian group, and let \[ \chi_{\mathcal G}(g,h)=\omega^{\Phi(g,h)} \] be a phase cocycle whose exponent $\Phi:\mathcal G\times\mathcal G\to \mathbb Z_m $ has a polynomial representation of degree \(d\) in finite abelian coordinates for \(g\) and \(h\). Then the cocycle phase oracle \[ U_{\chi,\mathcal G}:\ket{g}\ket{h}\longmapsto \chi_{\mathcal G}(g,h)\ket{g}\ket{h} \] can be synthesized using gates from the \(d\)th level of the Clifford hierarchy. 
\end{conjecture}

The cocycle exponent \(\Phi_{p,q}(x,y)\) from Clifford multiplication is quadratic in the blade coordinates, and explains why the phase oracle lives entirely within \(\mathcal C_2\). This is the motivating example for the conjecture, which asserts that the algebraic degree of the cocycle exponent influences the Clifford-hierarchy level needed to implement the associated diagonal phase oracle. If true, the conjecture would provide a systematic design principle for quantum phase oracles of twisted group algebras, thus offering a guide for the construction of efficient multiplication circuits and symmetry-based quantum algorithms beyond the Clifford case. A possible approach is to decompose the cocycle exponent into coordinate monomials and synthesize the diagonal phase gates individually. If every degree-\(d\) monomial phase belongs to the \(d\)th
level of the Clifford hierarchy, then the conjecture follows from the
commutativity of diagonal phase operators.

\subsection{Analysis}

The complexity results established above and below concern the quantum model of computation. As is the case with many quantum linear algebra algorithms, recovering all $N=2^n$ output coefficients classically requires $\Omega(N)$ measurements independent of the circuit complexity, an overhead not specific to Clifford multiplication but reflecting the general cost of extracting exponentially many classical values from a quantum state. The algorithm is best for computations in which the inputs are already available as quantum states or the output is consumed by subsequent quantum operations, expectation-value estimation, or measurement. 
 
State preparation presents a complementary challenge, as preparing an arbitrary amplitude-encoded state generally requires \(\Omega(N)\)
resources in the worst case. The relevant question is not whether arbitrary multivectors can be prepared efficiently, but whether Clifford multivectors can be classified in such a way that they have efficient quantum representations. Several important families already do. Basis blades require only computational-basis preparation, the uniform
multivector is obtained by applying \(H^{\otimes n}\) to \(\ket{0^n}\), stabilizer states prepared by Clifford circuits define structured amplitude-encoded multivectors in the blade basis, and tensor-product states give tensor-product multivectors. Identifying the broader class of efficiently preparable Clifford multivectors remains an open problem.

\subsubsection{Measurement}

The circuit prepares the desired product state in the trivial-character branch
$
N^{-1/2}|0^n\rangle\sum_zc_z|z\rangle.
$
Therefore
\[
p_0
=
N^{-1}\sum_z|c_z|^2
=
2^{-n}\|AB\|_2^2.
\]
If the output coefficient vector is normalized, then $p_0=2^{-n}$. Normalized output does not imply large success probability; it gives the generic exponential postselection penalty. Since $ c_z=\sum_xa_xb_{x\oplus z}\chi(x,x\oplus z), $ Young's inequality gives \[ \|AB\|_2 \le \|A\|_1\|B\|_2 \le 2^{n/2}\|A\|_2\|B\|_2. \]
For normalized inputs, this results in the universal bound
$
p_0\le 1,
$
with equality only if \(\|AB\|_2^2=2^n\).

This estimate alone does not characterize the typical success probability.
For random normalized dense coefficient vectors, the sums defining
\(c_z\) behave heuristically like random walks, suggesting \(\|AB\|_2^2=O(1)\) and 
\(p_0=O(2^{-n})\). Generic dense inputs are therefore expected to represent the least favorable regime for both classical and quantum Clifford multiplication. Many practical applications, however, possess additional algebraic structure in the input multivectors or the computational model. We analyze these distinct complexity regimes below.

\begin{table}[H]
\centering
\caption{Computational regimes for QCM. Effective complexity includes the cost of postselection or amplitude amplification when required.}
\label{tab:analysis_regimes}
\begin{tabular}{lll}
\toprule
Regime & Assumption & Effective complexity\\
\midrule
Dense inputs &
Typical $p_0=O(N^{-1})$ &
$O(N \log N)$ \\

Amplitude amplification &
$p_0\ge1/\mathrm{poly}(\log N)$ &
$O(\mathrm{poly}(\log N))$ \\

Structured inputs &
$k_A, k_B = \text{poly}(\log N)$ &
$\mathrm{poly}(k_A, k_B, \log N)$ \\
\bottomrule
\end{tabular}
\end{table}

\noindent
\paragraph{\textit{Amplitude amplification.}}
Amplitude amplification assumes efficient implementations of reflections about both the prepared input state and the success subspace. Reflection about the success subspace is implemented by a phase flip conditioned on the ancilla register identifying the desired branch, while reflection about the prepared state is done by conjugating the standard reflection about $|0\rangle$ with the state-preparation circuit. Since the Clifford multiplication circuit has linear depth, each Grover iteration incurs only linear circuit overhead. Suppose a family of inputs has success probability
$
p_0 \ge \frac{1}{\operatorname{poly}(n)}.
$
Then the postselected circuit can be promoted to an efficient state-preparation procedure. Since the Clifford multiplication circuit has depth $O(n)$, naive repetition prepares the desired branch with expected depth
$
O\!\left(\frac{\log N}{p_0}\right).
$
Amplitude amplification reduces the postselection overhead quadratically, giving expected depth
$
O\!\left(\frac{n}{\sqrt{p_0}}\right)
=
O\!\left(\frac{\log N}{\sqrt{p_0}}\right),
$
assuming efficient reflections about the prepared input state and the success subspace. Amplitude amplification makes polynomially small success probabilities efficient, but it does not convert exponentially small branch weight into a polynomial-depth algorithm.

\noindent
\paragraph{\textit{Support-adapted projection.}}
The factor \(2^{-n/2}\) is because the algorithm projects the first register onto the uniform superposition
\[
|+\rangle^{\otimes n}
=
2^{-n/2}\sum_{x\in(\mathbb Z_2)^n}|x\rangle,
\]
which averages uniformly over all \(2^n\) blade indices. If only a small subset of basis blades has nonzero coefficients, projecting onto the full superposition averages over many irrelevant basis blades. Now suppose that
\(
S_A=\operatorname{supp}(A)
\)
has cardinality \(s\), and that the normalized support state
\[
|S_A\rangle
=
s^{-1/2}\sum_{x\in S_A}|x\rangle
\]
can be prepared efficiently. Projecting onto \(|S_A\rangle\),
\[
s^{-1/2}
\sum_z
\left(
\sum_{x\in S_A}
a_xb_{x\oplus z}\chi(x,x\oplus z)
\right)
|z\rangle
=
s^{-1/2}\sum_z c_z|z\rangle,
\]
so the success probability becomes
\[
p_S
=
s^{-1}\|AB\|_2^2.
\]

When \(s=\operatorname{poly}(n)\) and
\(
\|AB\|_2^2\ge1/\operatorname{poly}(n),
\)
the postselection overhead is polynomial. Support-adapted projection replaces the ambient dimension \(2^n\) by the effective support size \(s\), but the improvement is conditional on the efficient preparation of the support
state \(|S_A\rangle\).

\paragraph{\textit{Sparse multivectors.}}
If
$
A=\sum_{i=1}^{k_A}a_ie_{x_i},
\qquad
B=\sum_{j=1}^{k_B}b_je_{y_j},
$
where
\(
k_A,k_B=\operatorname{poly}(n),
\)
then only \(k_Ak_B\) blade pairs contribute to the geometric product. Both the classical multiplication and the quantum state-preparation overhead are polynomial. Sparse multivectors provide a practically important setting in which the quantum multiplication primitive can be implemented without incurring exponential postselection overhead.

\subsection{Application}

\subsubsection{Property estimation and decision problems.}

The analysis identifies the cases in which quantum Clifford multiplication efficiently prepares the normalized product state
\[
|AB\rangle
=
\frac1{\|AB\|_2}\sum_z c_z|z\rangle.
\]
The most appropriate computational task is therefore not the reconstruction of every coefficient of the geometric product, but the efficient estimation of properties of the prepared state. As in the Harrow--Hassidim--Lloyd algorithm \cite{harrow2009quantum} and related quantum linear algebra algorithms, the product is accessed through efficiently measurable observables rather than complete classical output. 
This introduces a family of promise decision problems whose acceptance criterion depends on efficiently measurable properties of the Clifford product.

For example, letting
$
M_0=|0^n\rangle\!\langle0^n|
$
denote projection onto the scalar blade, 
\[
\langle AB|M_0|AB\rangle
=
\frac{|\langle AB\rangle_0|^2}{\|AB\|_2^2},
\]
one may verify whether the normalized scalar weight exceeds a prescribed threshold. Similarly, projection onto a fixed grade subspace determines whether the normalized grade-$k$ weight exceeds a threshold, while more generally any efficiently implementable projector onto blades, grades, parity sectors, ideals, or other algebraically defined subspaces defines an associated promise decision problem.

\begin{theorem}[Observable decision problems for QCM]
Let \(\{M_n\}\) be a uniform family of efficiently measurable observables.
If the normalized Clifford product state \(|AB\rangle\) is prepared with
inverse-polynomial success probability, then every promise decision problem
whose acceptance criterion is obtained by thresholding
\[
\langle AB|M_n|AB\rangle
\]
with an inverse-polynomial promise gap belongs to
\(\mathsf{BQP}\). Every promise problem with a polynomial-time many-one or
Turing reduction to estimating such expectation values belongs to
\(\mathsf{BQP}\).
\end{theorem}

Beyond observables defined directly by blades and grades, the same approach may be suited to decision problems from larger computational tasks whose solutions use Clifford algebras. Classical parameterized algorithms for Hamiltonian cycle, Steiner tree, feedback vertex set, and related graph problems on bounded-treewidth graphs employ
repeated Clifford multiplication through noncommutative subset convolution. If the relevant intermediate quantities can be encoded as efficiently measurable observables of the prepared product state, the quantum Clifford multiplication algorithm could be used within quantum algorithms for these decision problems.

\section{Outlook}
The geometric product is a dense, noncommutative bilinear operation, yet its noncommutativity is encoded entirely by a cocycle over an underlying Abelian group. Knowing this, the computational difficulty of Clifford multiplication is not in the noncommutativity, but in the efficient implementation of its cocycle twist. This invites the question of which group algebras, twisted group algebras, crossed products, and semisimple algebras have comparable quantum multiplication algorithms.

The work also points toward several practical directions. Since Clifford multiplication underpins rotors, spinors, versors, Clifford-valued differential operators, and geometric models of relativistic physics, an efficient multiplication primitive unlocks quantum algorithms devised natively in geometric algebra instead of through matrix representations. The close relationship between Clifford and matrix algebras also leaves open the extent to which these ideas extend to quantum matrix multiplication. Understanding these connections may uncover new quantum primitives for scientific computing, robotics, computer vision, spacetime simulation, and geometric numerical methods. Ultimately, the significance may lie not only in quantum Clifford multiplication itself, but in the wider computational direction it opens. The algorithm shows that harmonic methods need not be restricted to transforms; they may also provide efficient realizations of algebraic multiplication. Whether Clifford multiplication is an isolated example or the first member of a theory of harmonic quantum algebraic computation remains to be seen.

\bibliographystyle{plain}
\bibliography{references}

@book{grassmann1844,
  title     = {Die Lineale Ausdehnungslehre, ein neuer Zweig der Mathematik},
  author    = {Grassmann, Hermann},
  year      = {1844},
  publisher = {Otto Wigand},
  address   = {Leipzig}
}

@book{grassmann1862,
  title     = {Die Ausdehnungslehre: Vollst\"andig und in strenger Form bearbeitet},
  author    = {Grassmann, Hermann},
  year      = {1862},
  publisher = {Enslin},
  address   = {Berlin}
}

@article{clifford1878applications,
  author     = {Clifford, William Kingdon},
  title      = {Applications of {G}rassmann's Extensive Algebra},
  journal    = {American Journal of Mathematics},
  volume     = {1},
  number     = {4},
  pages      = {350--358},
  year      = {1878},
  publisher = {Johns Hopkins University Press}
}

@misc{Muchane2026Rotor,
  author = {Kagwe A. Muchane},
  title = {Rotor-Valued Orientation as Generalized Clifford Geometry},
  year = {2026},
  note = {Manuscript in preparation}
}

@article{Ablamowicz2009Complexity,
  author  = {R. Ab{\l}amowicz and B. Fauser},
  title   = {On Computational Complexity of Clifford Algebra},
  journal = {Journal of Mathematical Physics},
  volume  = {50},
  number  = {5},
  pages   = {053514},
  year    = {2009},
  doi     = {10.1063/1.3115453}
}

@book{Dorst2007GACS,
  author    = {Leo Dorst and Daniel Fontijne and Stephen Mann},
  title     = {Geometric Algebra for Computer Science: An Object-Oriented Approach to Geometry},
  publisher = {Morgan Kaufmann},
  address   = {Burlington, MA},
  year      = {2007},
  isbn      = {9780123749420}
}

@article{harrow2009quantum,
  title = {Quantum Algorithm for Linear Systems of Equations},
  author = {Harrow, Aram W. and Hassidim, Avinatan and Lloyd, Seth},
  journal = {Phys. Rev. Lett.},
  volume = {103},
  issue = {15},
  pages = {150502},
  numpages = {4},
  year = {2009},
  month = {Oct},
  publisher = {American Physical Society},
  doi = {10.1103/PhysRevLett.103.150502},
  url = {https://link.aps.org/doi/10.1103/PhysRevLett.103.150502}
}

@techreport{coppersmith1994approximate,
  title={An Approximate Fourier Transform Useful in Quantum Factoring},
  author={Coppersmith, Don},
  year={1994},
  institution={IBM Research},
  number={RC19642}
}

@book{Hestenes1966,
  title     = {Space-Time Algebra},
  author    = {Hestenes, David},
  year      = {2015},
  publisher = {Gordon and Breach},
  address   = {New York}
}

@book{HestenesSobczyk1984,
  title     = {Clifford Algebra to Geometric Calculus: A Unified Language for Mathematics and Physics},
  author    = {Hestenes, David and Sobczyk, Garret},
  year      = {1984},
  publisher = {Springer},
  series    = {Fundamental Theories of Physics},
  address   = {Dordrecht}
}

@inproceedings{wareham2005applications,
  title     = {Applications of Conformal Geometric Algebra in Computer Vision and Graphics},
  author    = {Wareham, Rich and Cameron, Jonathan and Lasenby, Joan},
  booktitle = {Computer Algebra and Geometric Algebra with Applications},
  pages     = {329--349},
  year      = {2005},
  publisher = {Springer}
}

@article{bayro2001geometric_neural,
  title     = {Geometric neural computing},
  author    = {Bayro-Corrochano, Eduardo J.},
  journal   = {IEEE Transactions on Neural Networks},
  volume    = {12},
  number    = {5},
  pages     = {968--986},
  year      = {2001},
  publisher = {IEEE}
}

@inproceedings{ruhe2023clifford,
  title     = {Clifford Group Equivariant Neural Networks},
  author    = {Ruhe, David and Brandstetter, Johannes and Forr{\'e}, Patrick},
  booktitle = {Advances in Neural Information Processing Systems (NeurIPS)},
  volume    = {36},
  pages     = {62922--62990},
  year      = {2023},
  publisher = {Curran Associates, Inc.}
}

@inproceedings{ruhe2023geometric,
  title     = {Geometric Clifford Algebra Networks},
  author    = {Ruhe, David and Gupta, Jayesh K. and Brandstetter, Johannes},
  booktitle = {Proceedings of the 40th International Conference on Machine Learning},
  series    = {Proceedings of Machine Learning Research},
  volume    = {202},
  pages     = {29461--29486},
  year      = {2023},
  publisher = {PMLR}
}

@article{buchholz2008clifford,
  title     = {On Clifford neurons and Clifford multi-layer perceptrons},
  author    = {Buchholz, Sven and Sommer, Gerald},
  journal   = {Neural Networks},
  volume    = {21},
  number    = {7},
  pages     = {925--935},
  year      = {2008},
  publisher = {Elsevier}
}

@inproceedings{brehmer2023geometric,
  title     = {Geometric Algebra Transformers},
  author    = {Brehmer, Johann and de Haan, Pim and Satorras, V{\'\i}ctor Garcia and Cohen, Taco},
  booktitle = {Advances in Neural Information Processing Systems (NeurIPS)},
  volume    = {36},
  year      = {2023}
}

@book{bayro2001geometric,
  title     = {Geometric Computing for Perception Action Systems: Concepts, Algorithms and Scientific Applications},
  author    = {Bayro-Corrochano, Eduardo},
  year      = {2001},
  publisher = {Springer}
}

@article{AlbuquerqueMajid2002,
  author={H. Albuquerque and S. Majid},
  title={Clifford algebras obtained by twisting of group algebras},
  journal={Journal of Pure and Applied Algebra},
  volume={171},
  number={2--3},
  pages={133--148},
  year={2002}
}

@book{Lounesto2001,
  author={Pertti Lounesto},
  title={Clifford Algebras and Spinors},
  edition={2},
  publisher={Cambridge University Press},
  year={2001}
}

@article{Castelazo2021QuantumConvolution,
  title = {Quantum algorithms for group convolution, cross-correlation, and equivariant transformations},
  author = {Castelazo, Grecia and Nguyen, Quynh T. and De Palma, Giacomo and Englund, Dirk and Lloyd, Seth and Kiani, Bobak T.},
  journal = {Phys. Rev. A},
  volume = {106},
  issue = {3},
  pages = {032402},
  numpages = {19},
  year = {2022},
  month = {Sep},
  publisher = {American Physical Society},
  doi = {10.1103/PhysRevA.106.032402},
  url = {https://link.aps.org/doi/10.1103/PhysRevA.106.032402}
}

@article{shor1997polynomial,
  author    = {Peter W. Shor},
  title     = {Polynomial-Time Algorithms for Prime Factorization and Discrete Logarithms on a Quantum Computer},
  journal   = {SIAM Journal on Computing},
  volume    = {26},
  number    = {5},
  pages     = {1484--1509},
  year      = {1997},
  doi       = {10.1137/S0097539795293172}
}

@inproceedings{kitaev1995quantum,
  author    = {A. Yu. Kitaev},
  title     = {Quantum Measurements and the Abelian Stabilizer Problem},
  booktitle = {Electronic Colloquium on Computational Complexity (ECCC)},
  year      = {1995},
  note       = {Expanded version in arXiv:quant-ph/9511026}
}

@article{lomont2004hidden,
  author  = {Chris Lomont},
  title   = {The Hidden Subgroup Problem -- Review and Open Problems},
  journal = {arXiv preprint quant-ph/0411037},
  year    = {2004}
}

@misc{muchane2025stateoperator,
  title={The State-Operator Clifford Compatibility: A Real Algebraic Framework for Quantum Information}, 
  author={Kagwe A. Muchane},
  year={2026},
  eprint={2512.07902},
  archivePrefix={arXiv},
  primaryClass={quant-ph},
  url={https://arxiv.org/abs/2512.07902}, 
}

@inproceedings{AlmanWilliams2021,
  author    = {Josh Alman and Virginia Vassilevska Williams},
  title     = {A Refined Laser Method and Faster Matrix Multiplication},
  booktitle = {Proceedings of the 2021 ACM-SIAM Symposium on Discrete Algorithms (SODA)},
  pages     = {522--539},
  year      = {2021},
  publisher = {Society for Industrial and Applied Mathematics},
  doi       = {10.1137/1.9781611976465.32},
  URL       = {https://epubs.siam.org/doi/10.1137/1.9781611976465.32}
}

\end{document}